\pdfoutput=1
\documentclass[runningheads,orivec,a4paper]{llncs}
\usepackage[utf8]{inputenc}
\usepackage[T1]{fontenc}
\usepackage{microtype}
\usepackage{mathtools}
\usepackage{float}
\usepackage{enumerate}
\usepackage{tikz}
\usepackage{subcaption}
\usepackage{amsfonts}

\usepackage{amsthm}

\PassOptionsToPackage{hyphens}{url}
\usepackage[hidelinks]{hyperref}
\usepackage{orcidlink}
\usepackage[nameinlink]{cleveref}
\hypersetup{
     colorlinks = true,
     allcolors = cblue
}
\newcommand*{\R}{\mathbb{R}}
\newcommand*{\N}{\mathbb{N}}

\newcommand*{\p}{\mathbf{p}}
\newcommand*{\s}{\mathbf{s}}
\newcommand*{\T}{[T]}
\DeclareMathOperator{\sgn}{sgn}

\usepackage{todonotes}

\newcommand*{\trlength}{1.3}
\newcommand*{\noderadius}{8pt}
\usetikzlibrary{arrows.meta, quotes, calc, fit, decorations.pathreplacing}
\tikzset{
    node distance={\trlength cm},
    vert/.style = {draw, circle, inner sep = 0pt, minimum size = 2*\noderadius, text depth=0.0ex, text height=1.25ex},
    every path/.style = {-{Latex[length=2mm]}},
    bid/.style = {{Latex[length=2mm]}-{Latex[length=2mm]}},
    every label/.append style={rectangle},
    fac/.style = {circle,fill,inner sep=1.5pt},
    unchosen/.style = {-, dashed},
    brace/.style = {decorate,decoration={brace,amplitude=7pt},-},
    bracelabel/.style = {midway,yshift=-0.55cm},
    invisible/.style = {rectangle, draw=none},
    box/.style = {rectangle, rounded corners, minimum height = 2*\noderadius + 2*\trlength cm},
    rect/.style = {thick, dotted}
}

\def\toffx{2.04}
\def\toffy{1.74}
\def\toffup{0.65}
\def\toffdown{0.71}
\def\toffleft{3}
\def\toffright{0.5}
\newcommand*{\boxes}[2]{
    \foreach \row / \price in {#2} {
        \coordinate(boxa\row) at (-\toffleft,      -\row*\toffy+\toffy+\toffup);
        \coordinate(boxb\row) at (#1*\toffx-\toffx+\toffright, -\row*\toffy+\toffy-\toffdown);
        \draw[rect] (boxa\row) rectangle (boxb\row);
        \node [below right, align=left] at (boxa\row) {\medmuskip=1mu\thinmuskip=1mu\thickmuskip=1mu$t=\row$\ifx\price\empty\else\\\medmuskip=1mu\thinmuskip=1mu\thickmuskip=1mu$p_\row=\price$\fi};
    }
}
\newcommand\makefromalph[1]{\number\numexpr`#1-`a\relax}
\newcommand{\increment}[1]{\expandafter\incrementhelper#1\relax}
\def\incrementhelper#1#2\relax{#1\the\numexpr#2+1\relax}
\newcommand*{\agent}[3]{
    \coordinate (coord#1#2) at (\makefromalph{#1}*\toffx, -#2*\toffy + \toffy);
    \node[vert] (#1#2) at (coord#1#2) {#3};
    
    \node[fac, label={[yshift=-5pt, xshift=3pt]above left:$#1_#2$}] (bat#1#2) at ($(coord#1#2) - (0.6, 0.5)$) {};
}
\newcommand*{\vertex}[3]{
    \coordinate (coord#1#2) at (\makefromalph{#1}*\toffx, -#2*\toffy + \toffy);
    \node[vert, label={[xshift=-5pt]below right:$#1_#2$}] (#1#2) at (coord#1#2) {#3};
}
\newcommand*{\charge}[3][]{
    \path (#2) edge["#3"', sloped, #1] (bat#2);
}
\newcommand*{\discharge}[3][]{
    \path (bat#2) edge["#3"', sloped, #1] (#2);
}
\newcommand*{\keep}[3][]{
    \path (bat#2) edge["#3", pos=0.19, #1] (bat\increment{#2});
}

\definecolor{cred}{HTML}{D81B60}
\definecolor{cblue}{HTML}{1E88E5}
\definecolor{cyellow}{HTML}{D09C00}
\definecolor{cgreen}{HTML}{5B8600}
\definecolor{cgray}{HTML}{AAAAAA}

\newtheorem{observation}{Observation}

\title{Social Welfare in Battery Charging Games}
\authorrunning{Krogmann, Lenzner, Skopalik, Sträubig}
\author{Simon Krogmann\inst{1}\orcidlink{0000-0001-6577-6756} \and Pascal Lenzner\inst{1}\orcidlink{0000-0002-3010-1019} \and Alexander Skopalik\inst{2}\orcidlink{0000-0002-4950-8708} \and \mbox{Tobias Sträubig\inst{3}\orcidlink{0009-0001-2362-5201}}}
\institute{Institute of Computer Science, University of Augsburg\\
\email{\{simon.krogmann, pascal.lenzner\}@uni-a.de}
\and Mathematics of Operations Research, University of Twente\\
\email{a.skopalik@utwente.nl}
\and Hasso Plattner Institute, University of Potsdam\\
\email{tobias.straeubig@student.hpi.de}}

\begin{document}

\maketitle

\setcounter{footnote}{0}

\begin{abstract}
   The recent rise of renewable energy produced by many decentralized sources yields interesting market design challenges for electrical grids. Balancing supply and demand in such networks is both a temporal and spatial challenge due to capacity constraints. The recent surge in the number of household-owned batteries, especially in regions with rooftop solar adoption, offers mitigation potential but often acts misaligned with grid-level objectives. In fact, the decision to charge or discharge a household-owned battery is a strategic choice by each battery owner governed by selfish incentives. This calls for an analysis from a game-theoretic point of view.

   We initiate this timely research direction by considering a game-theoretic setting where selfish agents strategically charge or discharge their batteries to increase their profit. In particular, we study a Stackel\-berg-like market model where a third party introduces price incentives, aiming to optimize renewable energy utilization while preserving grid feasibility. For this, we study the existence and the quality of equilibria under various pricing strategies. We find that the existence of equilibria crucially depends on the chosen pricing and that the obtained social welfare varies widely. This calls for more sophisticated market models and pricing mechanisms and opens up a rich field for future research in Algorithmic Game Theory on incentives in renewable energy networks.
\end{abstract}

\section{Introduction}
Renewable energy gained a large market share in recent years and it has even become the major energy source in some countries~\cite{eurostat}. With this transition, energy production shifted from being centralized and power-plant-based to various decentralized sources like household-owned rooftop solar panels or municipal wind turbines. However, the power grid that we currently rely on was built for centralized production and thus the new decentralized production model faces capacity problems when the energy supply from some regions meets the energy demand from other regions~\cite{hoogsteen17}. Moreover, besides these spatial challenges, also temporal challenges arise, since, for example, solar power is only available during the day but might be needed at night. One way to cope with this mismatch between production and consumption is adopting battery-based storage systems. This often happens on a small scale with household-owned batteries that can store a significant part of the home-produced energy to make it available for later use or to sell it to the grid. Within the last four years, the total battery storage in Europe went from 9.8 GWh in 2021 to 61.1 GWh in 2025. It is currently projected that this number will reach 400~GWh in 2029. Moreover, in the span of the last three years, 3 million new home batteries were connected to the power grid in Europe~\cite{home-batteries}.

This rise in the number of decentralized energy producers, who with their batteries can buy and sell energy at any given time, led to a significant growth of algorithmic trading of energy, which includes both algorithm-assisted training and pure algorithmic agents~\cite{algo-trading}, based on a survey of market participants in the Netherlands. Thus, battery charging decisions are made autonomously by automated agents. These agents try to sell surplus energy to the grid for a profit or they strategically buy cheap energy from the grid for later sale. All these agents interact with each other directly via the available capacities in the power grid and indirectly via the energy price.

However, all battery charging decisions by the agents must be supported by the available (and temporally changing) residual capacities of the power grid. Moreover, it must be ensured that the power grid can actually serve the demands of the customers (enough energy flow is available) and at the same time that selfish action of energy trading agents does not lead to an energy overload that would damage the valuable infrastructure (no capacities are exceeded by the energy flow). This yields a Stackelberg-like system where the Stackelberg leader, which could be the state-owned power grid administration or some other third-party authority, imposes a pricing scheme and possibly contractual penalties for maintaining power grid stability. The followers, i.e., the automatic trading agents, then try to maximize their profit by strategically buying and selling energy given the pricing scheme and avoiding the penalties. These actions of the agents can influence the overall energy flow in the network. For example, charging a battery at some time step with excess energy, i.e., energy that could not be consumed in this time step, allows for using this energy at a later time step to satisfy additional demand and thus increase the network's flow value. However, selfish battery charging decisions might not be optimal in this regard and thus might lead to a suboptimal utilization of excess energy.

The investigation of energy networks with strategic algorithmic energy trading agents opens up a rich field of study for the Algorithmic Game Theory community.
In this work, we set out to explore this setting. We consider a stylized model of an energy network with many algorithmic agents that are equipped with household-owned batteries and thus can selfishly buy and sell energy. Regarding the pricing, we assume that different time steps might have different prices that are known in advance. For example, these prices could be driven by the day-ahead market prices or they could be given by a third-party authority. Also, the algorithmic agents are limited by contractual penalties to ensure that the energy supply cannot be exceeded and to prevent power grid overload. Our main goal is to study the impact of the selfish algorithmic agents on the network's obtained energy flow. With this, we focus on an important aspect that, to the best of our knowledge, has not been studied before. This lays the foundation for further studies and more sophisticated models.


\subsection{Related Work}
Game theory has been widely applied to model strategic behavior in electricity markets. Recent surveys~\cite{8648326,10.3389/fenrg.2022.1009217,10.1063/5.0165108} provide a good overview of the work on this topic. There are multiple aspects in which our work differs from the related work.

A common feature in game-theoretic electricity market models is dynamic pricing, where prices are adjusted based on demand and supply to balance the grid efficiently~\cite{6266720,ehrhart2022congestion}. In contrast, our model assumes fixed energy prices per time step, where the agents' decisions do not affect the price. This removes complexities arising from the interaction between strategies and is in line with our assumption that each agent is a small household-based prosumer with very limited influence on the energy price.
Another key difference is how electricity is allocated to agents. Our model assumes a fixed amount of surplus electricity, which can be distributed to the agents in an arbitrary way. When the available electricity is fully purchased, strategy changes that would exceed the supply are never beneficial, causing inefficiencies that are central to our analysis. This contrasts many existing models that involve auction-based mechanisms or proportional allocation rules that attempt to distribute electricity fairly among competing agents ~\cite{10.1145/3396851.3397701,10.1145/3538637.3538843,8905496}.
Additionally, many works consider private electricity demand as the main driver of energy purchases. Some papers propose games in which agents seek to minimize the emerging costs by using energy storage~\cite{9026313,su11102763}. Our work, however, abstracts away private energy consumption and focuses purely on strategic energy storage and resale. In our model, electricity can be bought and sold at the same price, and households are reduced to their role as storage providers, making their utility a function of their trading strategy.

Other related work has considered some of the features of our model individually. In~\cite{ev-charging-game} a setting with multiple competing agents that strategically decide when it is best to charge an electric vehicle. The authors of~\cite{double-auctions} investigate market mechanisms for electricity markets with prosumers. 
Another model~\cite{collusion-prevention} investigates how to prevent collusion between prosumers which contrasts with some of our results where high cooperation is needed for better social welfare.
How electricity is routed within a time step is considered in~\cite{energy-flow}.

Market models for energy networks that are based on Stackelberg's model~\cite{Stackelberg} have been studied quite extensively, see e.g.~\cite{Liu17,Rahi19,cheng2024optimalpricingformulasmart} for a more game-theoretic view and \cite{Fochesato22,Dorahaki24} for a bi-level optimization approach.
Still, these works do not feature strategic battery charging and instead the main interaction of the agents is via the energy price.
Capacity constraints of the power grid and the resulting impact on the overall energy flow are ignored.

Finally, we emphasize that to the best of our knowledge, our paper is the first attempt to analyze the energy flow obtained by strategic battery charging agents within a capacity-constrained energy network from a game-theoretic point of view. Other electricity market models aim for a more realistic depiction of power grids while our approach is intentionally simplified to allow a focus on the energy flow. Much like abstract models like the Hotelling-Downs model~\cite{hotelling,downs} and its more recent variants~\cite{voronoigames,feldman-hotelling,Peters2018,alex-network-investment,ijcai-21}, which, despite their simplifications, provide valuable economic insights for spatial markets, our model is not designed to mirror real-world electricity markets directly. Instead, it serves as a theoretical foundation for studying strategic behavior in simplified settings and exploring potential model variations.

\subsection{Our Contribution}
We focus on capacitated energy networks with agents that are prosumers, i.e., that can produce and consume energy. Moreover, these prosumers have a battery that can be strategically used to charge excess energy at some time step and then later this energy can be sold back to the grid for a profit if there is demand. Our key contribution is that we explicitly model the impact on the obtained overall energy flow in the network - an aspect that was neglected in other models. 


Regarding the impact of the selfish battery charging decisions of the agents, we study the existence and the social welfare of equilibrium states, where the latter is the flow value of the induced maximum energy flow in the network. A higher value here means that more excess energy is utilized. We prove bounds on the price of anarchy (PoA) and the price of stability (PoS).

For instances with two time steps, we show that Nash equilibria exist for any pricing scheme. If prices are decreasing, then the PoS is unbounded, otherwise it is exactly~1. The PoA is trivially unbounded for decreasing and uniform prices but we show a tight bound of $2$ for ascending prices.
For the number of time steps $T$ larger than two, we show that for any instance, the uniform pricing scheme admits Nash equilibria, but that even for the very basic class of ascending pricing schemes, Nash equilibria might not exist.
For prices that increase over time, we show a high lower bound on the PoS of $\lfloor\frac{T}{2}\rfloor$. For prices that depend on the energy demand in each time step, we even get an infinite PoS. This is contrasted with a PoS of $1$ for uniform prices. Regarding the PoA, things are even worse. We show that even allowing cooperation among at most $T-2$ agents yields an unbounded PoA for any pricing scheme. If at least $T-1$ agents can cooperate, then the PoA is bounded by $T$, which is asymptotically tight.

Thus, overall we show that selfish battery charging decisions can severely impact the obtained energy flow in the network. Moreover, this holds for all pricing schemes and even under extensive agent cooperation. This shows that this aspect should also be studied in other models for energy trading networks.

\section{Model and Preliminaries}

\paragraph{\textbf{Electricity Grid.}}
We model the electricity network as a directed (host) graph ${H=(V,E)}$ where the set of vertices corresponds to households and the set of edges to power lines.
We consider $T\in\N$ time steps and denote $\T \coloneq \{1,\ldots, T\}$ and for each time step $t \in \T$, we let $H_t \coloneq (V_t, E_t)$ be a copy of $H$.
Each edge $e_t \in E_t$ has a capacity $\kappa(e_t) \in \R$.
Note that this allows grid capacities to vary over time, e.g., to model residual capacities.
Each node $v_t$ supplies or demands electricity given by the positive or negative value $d(v_t) \in \R$, respectively.

\paragraph{\textbf{Battery Agents.}}
For each household battery, there is a battery agent $b \in B$ located at a location~$v^b \in V$. For convenience, we will introduce for battery agent $b \in B$ an additional \emph{battery node} $b_t$ for each time step $t$ and \emph{battery edges} $(b_t,b_{t+1})$ for $t \in [T-1]$ with capacities $\kappa((b_t,b_{t+1}))$.
Furthermore, each battery node~$b_t$ is connected to node~$v^b_t$ by a bidirectional transaction edge which, for now, has infinite capacity. However, we later use these edges to model agents' decisions to charge or discharge the batteries.

\paragraph{\textbf{Energy Network Graph.}}
With that, we define the energy network graph $G=(V_G,E_G)$ as the union of all graphs $H_t$ for the individual time steps, the battery nodes and edges and the transaction edges, i.e.,
\begin{align*}
    V_G =& \textstyle\bigcup\limits_{t \in \T} \Bigl( {V_t} \cup \textstyle\bigcup\limits_{b \in B} b_t\Bigr)\\
    \text{and}\quad E_G =& \textstyle\bigcup\limits_{t \in \T} \Bigl( {E_t} \cup \textstyle\bigcup\limits_{b \in B} \{ (b_t, v^b_t), (v^b_t, b_t) \} \Bigr) \cup \textstyle\bigcup\limits_{t \in [T-1],b \in B} \{ (b_t,b_{t+1})\} \,.
\end{align*}
We give an example of such an energy network graph in \Cref{fig:structure}.
\begin{figure}[h]
\def\toffx{2.04}
\def\toffy{1.5}
\def\toffup{0.5}
\def\toffdown{0.6}
\def\toffleft{3}
\def\toffright{0.5}
    \centering
    \begin{tikzpicture}

        \agent{a}{1}{+1}
        \agent{b}{1}{}
        \agent{a}{2}{}
        \agent{b}{2}{-1}
        
        \path (a1) edge[bend left=10, "1", pos=0.37, inner sep=2pt] (b1);
        \path (a2) edge[bend left=10, "1", pos=0.37, inner sep=2pt] (b2);
        \path (b1) edge[bend left=10, "1", pos=0.63, inner sep=2pt] (a1);
        \path (b2) edge[bend left=10, "1", pos=0.63, inner sep=2pt] (a2);

        \charge[bid]{a1}{}
        \keep{a1}{2}
        \discharge[bid]{a2}{}
        \charge[bid]{b1}{}
        \keep{b1}{1}
        \discharge[bid]{b2}{}
        
        \boxes{2}{1/,2/}
    \end{tikzpicture}
    \caption{An instance of the charging game with two time steps. In this case, $H$ has two (unnamed) vertices connected by two edges. The supplies and demands are given inside the vertices. There are two battery agents $a$ and $b$. The vertical edges between the time steps are their battery edges and the network and battery capacities induced by $\kappa$ are given next to the edges.
    }
    \label{fig:structure}
\end{figure}
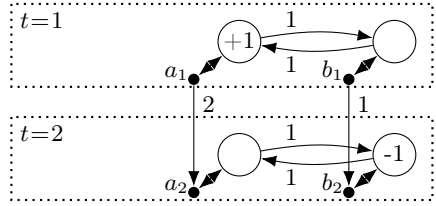%

\paragraph{\textbf{Strategies.}}
For each time step $t\in\T$, each agent decides on a positive or negative value $s_{b,t} \in \R$, i.e., how much to charge or discharge the battery, respectively.
Therefore, the strategy of an agent is a vector $s_b \in \R^{[T]}$. We denote a strategy profile by $\s \in \R^{B\times \T}$.
Thus, the agents precommit to a strategy for the whole game.
As a battery cannot charge above capacity or discharge below $0$, we limit $s_b$ to strategies that for each $t \in [T-1]$ satisfy
\[
0 \leq \sum_{z=1}^t s_{b,z} \leq \kappa((b_z, b_{z+1}))\text.
\]

\paragraph{\textbf{Electricity Flow.}}

For a strategy profile $\s$, we define the capacity profile $\kappa_\s$ on the graph $G$, by limiting the capacity of the transaction edges to the amount of electricity charged or discharged.
That is for a positive value $s_{b,t}$, i.e., charging, the capacity of the edge $(v^b_t, b_t)$ is $s_{b,t}$ and the reverse edge capacity $0$.
Conversely, for a negative value $s_{b,t}$, i.e., discharging, the capacity of the edge $(b_t, v^b_t)$ is $|s_{b,t}|$ and the reverse edge capacity $0$.
For all remaining edges~$e$, $\kappa_\s(e) = \kappa(e)$.
We refer by $G_\s$ to the graph~$G$ with the capacities induced by $\kappa_\s$.
With this at hand, we model the flow of electricity as a \emph{feasible} flow, ignoring transmission and conversion losses and the actual physics of power flows.
For ease of notation, we add a source node $x$ and a target node~$y$. The source $x$ is connected to each node $v$ with an edge $(x,v)$ with capacity $\kappa((x,v)) = \max(0, d(v))$. Similarly, for target node $y$, there are edges $(v,y)$ with capacity $\kappa((v,y)) = \max(0, -d(v))$.

Flow $f$ on the graph $G_\s$ assigns a value $f(e) \in \R_{\geq 0}$ to each edge $e$.
This flow has to satisfy the usual flow constraints:
\begin{enumerate}[(1)]
    \item \emph{Capacity Constraint}:
    $0 \leq f(e) \leq \kappa(e)$.
    \item \emph{Flow Conservation}:
    for each node $v \notin \{x, y\}$ we have
    \[
    \sum_{(v',v)\in E}f((v',v)) = \sum_{(v,v')\in E}f((v,v'))\text.
    \]
\end{enumerate}
The value $|f|$ of a flow is the amount of outgoing flow from $x$, i.e., \[
\sum_{v \in V}{f((x,v))} = \sum_{v \in V}{f((v, y))}\text.
\]
A flow in graph $G_\s$ is a maximum flow if it has the highest possible flow value $|f|$ of all flows in~$G_\s$.

\paragraph{\textbf{Admissible Strategies.}}
In our model, we only allow the charging of batteries that is possible without violating grid constraints or exceeding the supply. Likewise, electricity can only be discharged if it satisfies a demand.
Strategies that violate these constraints are prevented, e.g., by contractual penalties.

This is easily captured in our model: If for a graph $G_\s$, there is a maximum flow that does not fully utilize a transaction edge, then this transaction edge is unnecessary and causes either oversupply or overconsumption. Therefore, we call such an edge, and a strategy or strategy profile with such an edge, \emph{inadmissible}. Conversely, the absence of such an edge makes those admissible.

\begin{definition}[Admissibility]
    Let $\s$ be a strategy profile.
    A transaction edge $e \in C$ is called \emph{admissible} for $\s$ if every maximum flow on $G_\s$ saturates $e$.
    A strategy $s_b$ is called \emph{admissible} for $\s$ if every transaction edge $e \in \{ (b_t,b_{t+1}) \mid t \in [T-1]\}$ of agent $b$ is admissible and a strategy profile $\s$ is called \emph{admissible} if each of its strategies is admissible.
\end{definition}

Using maximum flows as the set of acceptable flows is based on the assumption that the operator wants to maximize the usage of excess energy that otherwise would be wasted.

\paragraph{\textbf{Prices and Payoffs.}}
The Stackelberg leader, e.g., the operator, decides on an electricity price profile $\p \in \R^{[T]}$ consisting of prices $p_t$ for each time step $t$, which holds for both buying and selling.

Therefore, the payoff an agent receives in case of admissibility is the sum of payments for buying and selling electricity of all time steps, i.e., $-\sum_{t\in\T}s_{b,t}p_t$.
Note that there is no discount on locally produced electricity.
However, agents with an inadmissible strategy are penalized which we model by a utility of $-\infty$ for the sake of simplicity.
Therefore, the utility function for battery agent $b$ is
\[
    u_b(\s, \p) \coloneq \begin{cases}
        -\infty,& s_b \text{ is not admissible,} \\
        -\sum_{t\in\T}s_{b,t}p_t, &\text{else.}
    \end{cases}
\]

\paragraph{\textbf{Equilibria.}}
In this work, we consider pure Nash equilibria and $k$-strong Nash equilibria.
For a given price profile $\p$, a strategy profile $\s$ is a Nash equilibrium\footnote{From now on, we omit \emph{pure} but still refer to pure Nash equilibria.} if for no agent $b$, there exists an alternative strategy $s'_b$ such that $u_b((s'_b, \s_{-b}), \p) > u_b(\s, \p)$ where $\s_{-b}$ is the strategy profile of all agents except $b$.

A strategy profile $\s$ is a $k$-strong equilibrium if there is no coalition $C$ of $k$ agents with a partial strategy profile $\s'_C \in \R^C$ such that $u_b((\s'_C, \s_{-C}), \p) > u_b(\s, \p)$ for all $b \in C$ where $\s_{-C} \in \R^{B \setminus C}$ is the strategy profile of the set of battery agents outside of $C$.

\paragraph{\textbf{Efficiency and Social Welfare.}}
Since the system objective is the utilization of excess energy, we measure the welfare of a strategy profile $\s$ as the value of the maximal flow:
\[
W(\s)= \max_{\text{feasible flow } f \text{ in } G_\s} |f|.
\]

As we want to quantify the (in-)efficiency of equilibria in comparison to an optimal solution of a central planner, we employ the help of the standard ratios known as the price of anarchy and price of stability. As we later seek to study certain classes of price profiles, we define them for a given set $P$ of price profiles.
For an instance $G$, let $\textup{OPT}(G)$ be the strategy profile with the highest welfare.
For an instance $G$ and a price function~$p$, let $\textup{bestNE}(G,p)$ and $\textup{worstNE}(G,p)$ be the flow maximizing and minimizing Nash equilibria, respectively.
Hence, the price of stability and anarchy for a class of price profiles $P$ are
\[
\textup{PoA}_P \coloneq\sup_{G} \inf_{p \in P}\frac{W(\textup{OPT}(G))}{W(\textup{worstNE}(G,p))}
\ \ \text{and}\ \ 
\textup{PoS}_P \coloneq\sup_{G} \inf_{p \in P} \frac{W(\textup{OPT}(G))}{W(\textup{bestNE}(G,p))}\text.
\]
Note that this implicitly allows the prices to depend on the instance $G$.
We also consider the $k$-strong variants of the price of anarchy PoA$_P^{k\text{-strong}}$ and stability PoS$_P^{k\text{-strong}}$. In this case, the ratios refer to the best and worst $k$-strong equilibria.

\paragraph{\textbf{Price Profiles.}} We define some special classes of price profiles:
\begin{description}
    \item{\emph{any} (ANY):}
        The set of all price profiles $p$
    \item{\emph{ascending} (ASC):}
        $p$ is ascending in time for all $G$.
    \item{\emph{descending} (DESC):}
        $p$ is descending in time for all $G$.
    \item{\emph{uniform} (UNI):}
        $p$ is identical for all time steps.
    \item{\emph{supply-sign} (SIGN):}
        for each $t \in \T$: $\sgn(p_t) = -\sgn\Bigl(\sum\limits_{v_t\in V_t}d(v_t)\Bigr)$
\end{description}
That means that supply-sign functions have negative prices if there is too much electricity in a time step and positive prices if there is not enough.

\section{Preliminary Observations}
\label{sec:strategies}

Given a strategy profile $\s$ and the corresponding graph $G_\s$, it is possible that not all transactions of the agents, i.e., all planned charging and discharging operations, can be fulfilled.
If a strategy profile is inadmissible, i.e., some transaction edge in $G_\s$ is not saturated in every maximum flow, then an individual transaction might not be fulfilled depending on the specific flow.
We discourage such ambiguities resulting from an agent buying electricity that is already in use or selling it when there is an oversupply.
However, in such a scenario, it is not clear who is at fault because, in some cases, multiple agents can change their strategies to get back to an admissible strategy profile.
Our definitions of utility and admissibility are designed to punish the set of agents responsible for this because a transaction of them is not necessary for the maximum flow.
We give an example of this in \Cref{fig:flow-example}, which also shows strategies that augment flow.

\begin{figure}[h]

\def\toffx{2.04}
\def\toffy{1.6}
\def\toffup{0.55}
\def\toffdown{0.66}
\def\toffleft{1.9}
\def\toffright{0.5}
    \begin{subfigure}[t]{0.48\columnwidth}
        \centering
        \resizebox{\textwidth}{!}{%
        \begin{tikzpicture}
            \agent{a}{1}{+1}
            \agent{b}{1}{}
            \agent{a}{2}{}
            \agent{b}{2}{-1}
            \agent{c}{1}{+1}
            \agent{c}{2}{-1}
            \path (a1) edge[bend left =10, "1", pos=0.37, inner sep = 2pt] (b1);
            \path (a2) edge[bend left =10, "1", pos=0.37, inner sep = 2pt] (b2);
            \path (b1) edge[bend left =10, "1", pos=0.63, inner sep = 2pt] (a1);
            \path (b2) edge[bend left =10, "1", pos=0.63, inner sep = 2pt] (a2);
            
            \keep{a1}{1}
            \keep{b1}{1}

            \charge{c1}{1}
            \keep{c1}{1}
            \discharge{c2}{1}
            
            \boxes{3}{1/1,2/2}
        \end{tikzpicture}}
        \caption{The graph $G_\s$ for strategies $s_a=s_b=(0,0)$ and $s_c=(+1,-1)$. The profile is admissible, since there are no inadmissible transaction edges.\\ }
    \end{subfigure}
    \hfill
    \begin{subfigure}[t]{0.48\columnwidth}
        \centering
        \resizebox{\textwidth}{!}{%
        \begin{tikzpicture}
            \agent{a}{1}{+1}
            \agent{b}{1}{}
            \agent{a}{2}{}
            \agent{b}{2}{-1}
            \agent{c}{1}{+1}
            \agent{c}{2}{-1}
            \path (a1) edge[bend left=10, "1", pos=0.37, inner sep = 2pt] (b1);
            \path (a2) edge[bend left=10, "1", pos=0.37, inner sep = 2pt, cblue, dashed] (b2);
            \path (b1) edge[bend left=10, "1", pos=0.63, inner sep = 2pt] (a1);
            \path (b2) edge[bend left=10, "1", pos=0.63, inner sep = 2pt] (a2);
            
            \keep{b1}{1}

            \charge[cblue, dashed]{a1}{1}
            \keep[cblue, dashed]{a1}{1}
            \discharge[cblue, dashed]{a2}{1}
            
            \charge[cblue, dashed]{c1}{1}
            \keep[cblue, dashed]{c1}{1}
            \discharge[cblue, dashed]{c2}{1}
            
            \boxes{3}{1/1,2/2}
        \end{tikzpicture}}
        \caption{The graph $G_\s$ if agent $a$ changes her strategy to $s_a=(1,-1)$ with the only maximum flow indicated in dashed blue. All transaction edges are admissible.}
    \end{subfigure}
    \begin{subfigure}[t]{0.48\columnwidth}
        \centering
        \resizebox{\textwidth}{!}{%
        \begin{tikzpicture}
            \agent{a}{1}{+1}
            \agent{b}{1}{}
            \agent{a}{2}{}
            \agent{b}{2}{-1}
            \agent{c}{1}{+1}
            \agent{c}{2}{-1}
            \path (a1) edge[bend left=10, "1", pos=0.37, inner sep = 2pt, cred, ultra thick] (b1);
            \path (a2) edge[bend left=10, "1", pos=0.37, inner sep = 2pt, cblue, dashed] (b2);
            \path (b1) edge[bend left=10, "1", pos=0.63, inner sep = 2pt] (a1);
            \path (b2) edge[bend left=10, "1", pos=0.63, inner sep = 2pt] (a2);

            \charge[cblue, dashed]{a1}{1}
            \keep[cblue, dashed]{a1}{1}
            \discharge[cblue, dashed]{a2}{1}
            
            \charge[cred, ultra thick]{b1}{1}
            \keep[cred, ultra thick]{b1}{1}
            \discharge[cred, ultra thick]{b2}{1}
            
            \charge{c1}{1}
            \keep{c1}{1}
            \discharge{c2}{1}
            
            \boxes{3}{1/1,2/2}
        \end{tikzpicture}}
        \caption{The graph $G_\s$ if agent $b$ also plays $s_b=(1,-1)$. While every maximum flow uses the edges of $c$, giving her utility $1$, we now have multiple distinct maximum flows on the left side.
        Each particle could either use the thick red or the dashed blue path, so the strategies of agents $a$ and $b$ are inadmissible, leading to utility $-\infty$.}
        \label{fig:example-unsatisfiable}
    \end{subfigure}
    \hfill
    \begin{subfigure}[t]{0.48\columnwidth}
        \centering
        \resizebox{\textwidth}{!}{%
        \begin{tikzpicture}
            \agent{a}{1}{+1}
            \agent{b}{1}{}
            \agent{a}{2}{}
            \agent{b}{2}{-1}
            \agent{c}{1}{+1}
            \agent{c}{2}{-1}
            \path (a1) edge[bend left=10, "1", pos=0.37, inner sep = 2pt, cred, ultra thick] (b1);
            \path (a2) edge[bend left=10, "1", pos=0.37, inner sep = 2pt] (b2);
            \path (b1) edge[bend left=10, "1", pos=0.63, inner sep = 2pt] (a1);
            \path (b2) edge[bend left=10, "1", pos=0.63, inner sep = 2pt] (a2);
            
            \keep{a1}{1}
            
            \charge[cred, ultra thick]{b1}{1}
            \keep[cred, ultra thick]{b1}{1}
            \discharge[cred, ultra thick]{b2}{1}
            
            \charge[cred, ultra thick]{c1}{1}
            \keep[cred, ultra thick]{c1}{1}
            \discharge[cred, ultra thick]{c2}{1}
            
            \boxes{3}{1/1,2/2}
        \end{tikzpicture}}
        \caption{The graph $G_\s$ for a second Nash equilibrium with agent $b$ playing $s_b=(1,-1)$, but this time $s_a=(0,0)$. With only one maximum flow (thick red), this profile is admissible and all agents receive non-negative utility.}
    \end{subfigure}
    \caption{For a graph $G$, the subfigures contain the graphs $G_\s$ for several strategy profiles $\s$ and the corresponding maximum flows.}
    \label{fig:flow-example}
\end{figure}
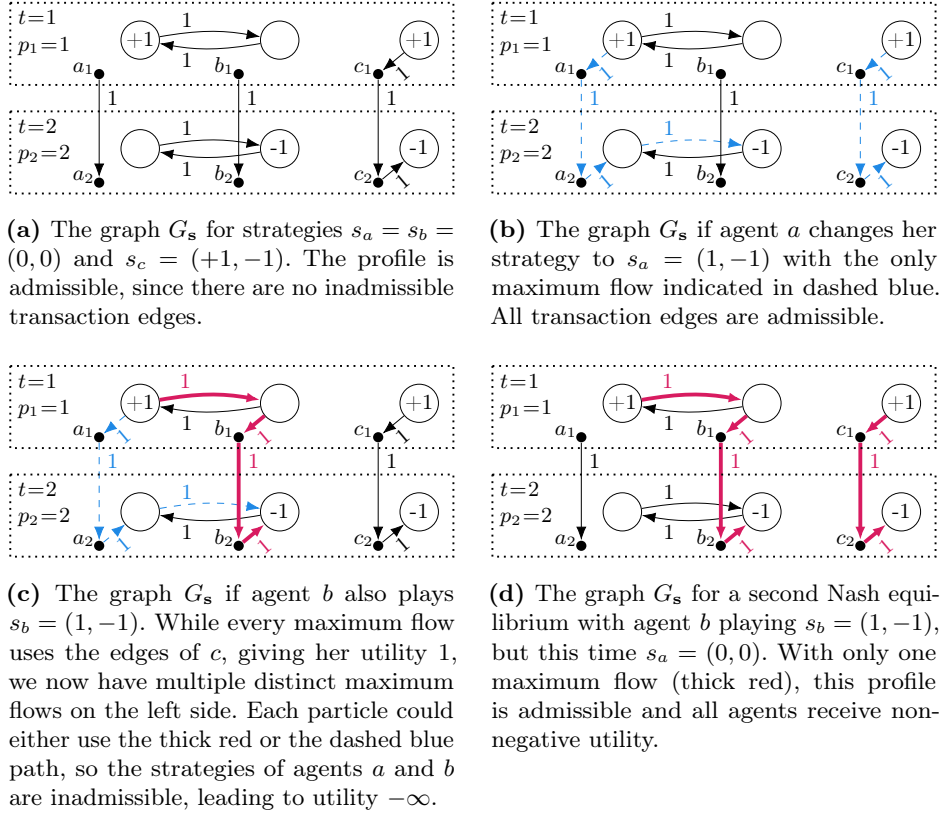

In \Cref{fig:example-unsatisfiable}, we show why not all agents are punished in an inadmissible strategy profile:
Agent $c$ cannot fix the inadmissibility of the profile by changing strategies and in fact, any change in her strategy wastes even more energy.
Which agents exactly have admissible strategies is computable in polynomial time.

\begin{theorem}[Utilities in Polynomial Time]
   For a given strategy and price profile, admissibility and the agents' utilities can be computed in polynomial time.
\end{theorem}
\begin{proof}
    The admissibility of a transaction edge $e$ can be computed by lowering its capacity by $\epsilon>0$ and computing a maximum flow for the new instance.
    If the maximum flow is still the same, the edge is not admissible because there is an alternative flow that does not saturate $e$.
    The agents' utilities can be trivially obtained in polynomial time from the computed maximum flow.
\end{proof}

Each agent can play the strategy $(0)_{t\in\T}$ to get a utility of $0$. This means that the agent does not buy any electricity and therefore none of her battery edges are inadmissible. Thus, no agent can have negative utility in any equilibrium.

\begin{observation}[Non-Negative Utilities]
    \label{obs:posutil}
    In any Nash equilibrium, agents have non-negative utility.
\end{observation}
This implies that every Nash equilibrium is an admissible strategy profile because in an inadmissible profile, at least one agent receives utility $-\infty$.

\section{Two Time Steps}

As a warm-up, we investigate the case of two time steps.
This is a relatively simple case because there are only three price functions to look at:
Ascending, descending and equal prices over the two steps.
The absolute difference does not matter, because it only scales the utilities.
Of these, only ascending prices are interesting because for descending prices, the only equilibrium is for every agent to do nothing.
For equal prices, every state is an equilibrium.
Having said that, we show that even for ascending prices, there always exists an equilibrium.

\begin{theorem}[$T=2$ Equilibrium Existence]
    For an instance $G$ with $T=2$ time steps, there exists a Nash equilibrium for every price profile. \label{thm:existence_T=2}
\end{theorem}
\begin{proof}
    If $p_1 > p_2$, i.e., the price decreases, the only equilibrium is the strategy profile where no agent buys or sells electricity.
    For $p_1 = p_2$, all admissible strategy profiles are Nash equilibria, since the agents are indifferent between strategies.

    For the case of $p_1 < p_2$, we show that the maximum flow $f$ in $G$ yields a Nash equilibrium $\s$:
    For any battery edge $e = (b_1,b_2)$, we set $s_{b} = (f_e, -f_e)$.
    This is admissible, since the battery edges in $\kappa_\s$ form a cut whose size equals the value of the maximum flow.
    Thus, and because of the increasing price for an agent $b$, buying less electricity decreases the utility.
    An agent $b$ buying more electricity causes the edge $b_t$ to be inadmissible because $f$ is a maximum flow that does not fully utilize the transaction edges of agent~$b$.
\end{proof}

As a corollary, the price of stability is $1$ for ascending and uniform price functions. In contrast, the price of stability is $\infty$ for descending prices.
\begin{theorem}[$T=2$ Price of Stability] \label{thm:PoS_T=2}
    For $T=2$ time steps, the prices of stability are $PoS_\textup{UNI}=PoS_\textup{ASC}=1$ and $PoS_\textup{DESC}=\infty$.
\end{theorem}
\begin{proof}
$PoS_\textup{UNI}=PoS_\textup{ASC}=1$ follows from the proof of \Cref{thm:existence_T=2}, since all employed battery edges of the optimum flow are profitable (or at least neutral) for the involved agents. Thus, they cannot improve their utility by not charging their battery. Also, since the flow has maximum value, these agents also cannot change to some other strategy, as this would yield inadmissible states.

For descending prices, the situation is different. If we consider the instance depicted in \Cref{fig:structure}, then no agent is willing to charge her battery since a later sale would result in a negative utility. Thus, the flow value is $0$. In contrast, the optimum flow would utilize one of the batteries, achieving a flow value of $1$.
\end{proof}
However, as we see next, worse Nash equilibria exist. In fact, we show that the PoA is at least $2$ for any price function.
\begin{theorem}[$T=2$ Price of Anarchy]
\label{thm:T2PoAlb}
    For $T=2$ time steps, the price of anarchy is $PoA_p \geq 2$ for any price profile $p$.
\end{theorem}
\begin{proof}
Consider the electricity network $G$ shown in \Cref{fig:t2poa} where each agent has a battery capacity of $1$.
We show that any price function $p$ has $PoA_p \geq 2$.
The strategy profile where agents $a$ and $e$ store one unit of electricity each yields a flow of $2$ (dashed blue in \Cref{fig:t2poa}) and thus it is optimal.

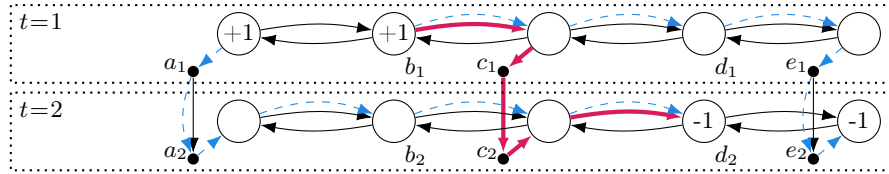
\begin{figure}[h]
\def\toffx{2.04}
\def\toffy{1.16}
\def\toffup{0.38}
\def\toffdown{0.65}
\def\toffleft{3}
\def\toffright{0.5}
    \centering
    \begin{tikzpicture}
        \agent{a}{1}{+1}
        \vertex{b}{1}{+1}
        \agent{c}{1}{}
        \vertex{d}{1}{}
        \agent{e}{1}{}
        \agent{a}{2}{}
        \vertex{b}{2}{}
        \agent{c}{2}{}
        \vertex{d}{2}{-1}
        \agent{e}{2}{-1}
        
        \path (a1) edge[bend left=10] (b1);
        \path (b1) edge[bend left=10] (a1);
        \path (b1) edge[bend left=10, ultra thick, cred] (c1);
        \path (c1) edge[bend left=10] (b1);
        \path (c1) edge[bend left=10] (d1);
        \path (d1) edge[bend left=10] (c1);
        \path (d1) edge[bend left=10] (e1);
        \path (e1) edge[bend left=10] (d1);
        
        \charge[ultra thick, cred]{c1}{}
        \keep[ultra thick, cred]{c1}{}
        \discharge[ultra thick, cred]{c2}{}
        
        \path (a2) edge[bend left=10] (b2);
        \path (b2) edge[bend left=10] (a2);
        \path (b2) edge[bend left=10] (c2);
        \path (c2) edge[bend left=10] (b2);
        \path (c2) edge[bend left=10, ultra thick, cred] (d2);
        \path (d2) edge[bend left=10] (c2);
        \path (d2) edge[bend left=10] (e2);
        \path (e2) edge[bend left=10] (d2);
        
        \charge[dashed, cblue]{a1}{}
        \keep{a1}{}
        \keep[dashed, cblue, bend right=20]{a1}{}
        \discharge[dashed, cblue]{a2}{}
        \path (a2) edge[cblue, dashed, bend left=20] (b2);
        \path (b2) edge[cblue, dashed, bend left=20] (c2);
        \path (c2) edge[cblue, dashed, bend left=20] (d2);
        
        \charge[dashed, cblue]{e1}{}
        \keep{e1}{}
        \keep[dashed, cblue, bend right=20]{e1}{}
        \discharge[dashed, cblue]{e2}{}
        \path (b1) edge[cblue, dashed, bend left=20] (c1);
        \path (c1) edge[cblue, dashed, bend left=20] (d1);
        \path (d1) edge[cblue, dashed, bend left=20] (e1);
        
        \boxes{5}{1/,2/}
    \end{tikzpicture}
    \caption{The graph $G_\s$ where only agent $c$ buys/sells electricity with capacity $1$ on all edges. The thick red edges mark the maximum flow in $G_\s$. Profile $\s$ is a Nash equilibrium with social welfare $1$. However, an alternative strategy profile enabling the dashed blue flow has a social welfare of~$2$.}
    \label{fig:t2poa}
\end{figure}

Let $p$ be an arbitrary price function. We make a case distinction and first consider the case that $p_1 \geq p_2$, i.e., the cost is not increasing.
In this case, there is an equilibrium with social welfare $0$ where no agent buys any electricity. Thus, in this case, the social welfare ratio with the optimum is infinity. For $p_1 < p_2$, we claim that strategy profile $\s$ with $s_{c}=(1,-1)$ and no other agents buying electricity is a Nash equilibrium with social welfare $1$.
Agent $c$'s only possible deviation of buying less electricity decreases utility. For any other agent $x \neq c$ buying more electricity, there is a minimum cut of size $1$ between the supplies and demands, either in the form of $(b_1, v^c_1)$ or $(v^c_2, d_2)$.
Thus, if $x$ deviates, it receives utility $-\infty$.
Therefore, the price of anarchy is at least $2$.
\end{proof}
We will later see that the lower bound for two time steps from \Cref{thm:T2PoAlb} is tight for ascending prices (see \Cref{thm:poa-t-strong} in \Cref{sec:multiple}).

\section{More Than Two Time Steps}
\label{sec:multiple}


In this section, we study more than two time steps. We show that for a given instance, there always exists a price profile that induces a Nash equilibrium.
However, there are price profiles that do not admit any pure Nash equilibrium. We also show that the price of stability is high for ascending prices, while the price of anarchy is high for all price profiles.

\subsection{Equilibrium Existence}
We start with the observation that for any given instance, there always exists a price profile that admits a Nash equilibrium. This price profile is uniform pricing. With this, each agent's payoff from charging and discharging evaluates to $0$. Thus,
any admissible strategy profile is a Nash equilibrium.
\begin{observation}[Equilibrium Existence]
    For any given instance, there always exists a price profile that admits a Nash equilibrium.
\end{observation}

However, for a given price profile, equilibria might not exist.

\begin{theorem}[Equilibrium Non-Existence]
\label{thm:no-ne}
    For a given instance and price profile, the charging game does not necessarily admit a Nash equilibrium.
\end{theorem}
\begin{proof}
    Consider the instance and price profile $\p$ in \Cref{fig:no-ne} with capacity $2$ on all edges and batteries.
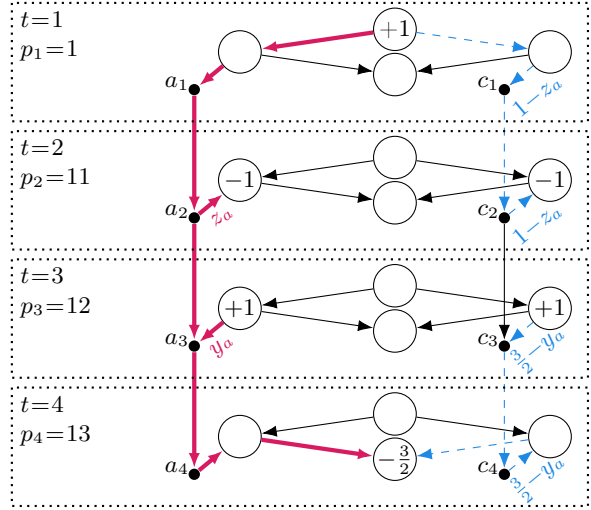
\begin{figure}[h]
\def\toffx{2.04}
\def\toffy{1.7}
\def\toffup{0.65}
\def\toffdown{0.92}
\def\toffleft{3}
\def\toffright{0.5}
    \centering
    \begin{tikzpicture}
        \agent{a}{1}{}
        \node[vert] (bo1) at (\toffx, -0*\toffy +0.3) {+1};
        \node[vert] (bi1) at (\toffx, -0*\toffy -0.3) {};
        \agent{c}{1}{}
        \agent{a}{2}{$-1$}
        \node[vert] (bo2) at (\toffx, -1*\toffy +0.3) {};
        \node[vert] (bi2) at (\toffx, -1*\toffy -0.3) {};
        \agent{c}{2}{$-1$}
        \agent{a}{3}{$+1$}
        \node[vert] (bo3) at (\toffx, -2*\toffy +0.3) {};
        \node[vert] (bi3) at (\toffx, -2*\toffy -0.3) {};
        \agent{c}{3}{$+1$}
        \agent{a}{4}{}
        \node[vert] (bo4) at (\toffx, -3*\toffy +0.3) {};
        \node[vert] (bi4) at (\toffx, -3*\toffy -0.3) {$-\frac32$};
        \agent{c}{4}{}
        
        \path (bo1) edge[cred, ultra thick] (a1);
        \path (bo1) edge[cblue, dashed] (c1);
        \path (a1) edge (bi1);
        \path (c1) edge (bi1);

        \foreach \i in {2,3} {
            \path (bo\i) edge (a\i);
            \path (bo\i) edge (c\i);
            \path (a\i) edge (bi\i);
            \path (c\i) edge (bi\i);
        }
        
        \path (bo4) edge (a4);
        \path (bo4) edge (c4);
        \path (a4) edge[cred, ultra thick] (bi4);
        \path (c4) edge[cblue, dashed] (bi4);
        
        \charge[cred, ultra thick]{a1}{}
        \keep[cred, ultra thick]{a1}{$ $}
        \discharge[cred, ultra thick]{a2}{$z_a$}
        \keep[cred, ultra thick]{a2}{$ $}
        \charge[cred, ultra thick]{a3}{$y_a$}
        \keep[cred, ultra thick]{a3}{$ $}
        \discharge[cred, ultra thick]{a4}{}
        
        \charge[cblue, dashed, inner sep=6pt]{c1}{\medmuskip=0mu$1-z_a$}
        \keep[cblue, dashed]{c1}{}
        \discharge[cblue, dashed, inner sep=6pt]{c2}{\medmuskip=0mu$1-z_a$}
        \keep{c2}{}
        \charge[cblue, dashed,pos=0.4]{c3}{\medmuskip=0mu$\frac32-y_a$}
        \keep[cblue, dashed]{c3}{}
        \discharge[cblue, dashed,pos=0.6]{c4}{\medmuskip=0mu$\frac32-y_a$}

        \boxes{3}{1/1,2/11,3/12,4/13}
        
    \end{tikzpicture}
    \caption{Graph $G$ of an instance with no Nash equilibrium, all capacities are $2$.
    In \Cref{thm:no-ne}, we consider equilibrium candidates in which $z_a$ (or $z_c$) is discharged in time step~2 by agent $a$ (or agent $c$) and $y_a$ (or $y_c$) is charged in time step~3,
    visualized with its intended flow in red for agent~$a$.
    For both agents, we consider a deviation to a strategy as depicted for agent~$c$ with the blue dashed strategy $(1-z_a, -1+z_a, \frac32-y_a, -\frac32+y_a)$.}
    \label{fig:no-ne}
\end{figure}
Assume towards a contradiction that there is a Nash equilibrium $\s$.
For strategy profile $\s$, let $z_a$ and $z_c$ denote the discharge of agents $a$ and $c$ in time step 2, respectively. Let $y_a$ and $y_c$ denote the charge of agents $a$ and $c$ in time step 3, respectively. While this does not fully characterize $\s$, these parameters are sufficient to show that $\s$ cannot exist. Observe that $z_a + z_c \le 1$ and $y_a, y_c \le 1$ as otherwise $\s$ is not admissible, contradicting \Cref{obs:posutil}.

It will be useful to bound the total utility in $\s$. To this end, observe that the total charge in time step 1 is $1$ and the total discharge in time step 4 is $\frac32$ in any Nash equilibrium. The values cannot be higher as this violates admissibility. A lower value would allow an improvement by one of the agents transporting some more flow between time steps 1 and 2 or 3 and 4, respectively.
With that, we can bound the total payoff of agents $a$ and $c$ by
\begin{align}
    u_a(\s) + u_b(\s) = 10 (z_a + z_c) + 12 (1 - z_a - z_c) +  1 (y_a + y_c). \label{ne}
\end{align}
Furthermore, if $y_a > 0$ and $z_a > 0$, agent $a$ could improve by decreasing both values by some $\epsilon$, resulting in a gain of $\epsilon$.
Thus, and with the same argument for agent $c$, we have
\begin{align}
     z_a > 0 \Rightarrow y_a = 0 \text{ and } z_c > 0 \Rightarrow y_c = 0 . \label{imply}
\end{align}
We now consider for agent $a$ the new strategy (and analogously for agent $c$)
\begin{align*}
s'_a=&\left(1- z_c, - 1 + z_c, \tfrac32 - y_c, - \tfrac32 + y_c\right)\text{ and}\\
s'_c=&\left(1- z_a, - 1 + z_a, \tfrac32 - y_a, - \tfrac32 + y_a\right)\text.
\end{align*}
We observe that the profile $(s'_a,s_c)$ yields a maximal flow of $\frac52$, which is the best possible for this network, by sending in total $1$ unit from time step 1 to 2 and $\frac{3}{2}$ units from time step 3 to 4. Thus, the strategy $s'_a$ is admissible for the deviating agent $a$ since all her transaction edges are necessary for the maximal flow. However, it might make agent $c$'s strategy no longer admissible if it transports some flow between time steps 3 and 4.
The utility of agent $a$ (and analogous for agent $c$) after deviating to strategy $s'_a$ (or $s'_c$, respectively) is
\begin{align}
u_a((s'_a,s_c))=10 (1 - z_c) + \tfrac32 - y_c \text{ and } u_c((s_a,s'_c))=10 (1 - z_a) + \tfrac32 - y_a\text. \label{dev}
\end{align}
By the fact that the profile $\s$ is a Nash equilibrium, we get
$u_a(\s) \ge u_a((s_a',s_c))$ and $u_c(\s) \ge u_c((s_a,s_c'))$ and hence
\begin{equation}
u_a(\s) + u_c(\s)  \ge u_a((s_a',s_c)) + u_c((s_a,s_c'))\text.
\label{NEsum}\end{equation}
By substituting (\ref{NEsum}) with (\ref{ne}) and (\ref{dev}) we get
\begin{align}
  10 (z_a + z_c) + 12 (1 - z_a - z_c) + ( y_a + y_c) &\ge 10 (2-z_a-z_c) + 3 - y_a-y_c\nonumber\\
  -2 (z_a + z_c) + 12  + y_a + y_c &\ge 23 - 10 (z_a+z_c)  - y_a-y_c \nonumber \\
 8 (z_a + z_c) + 2 (y_a + y_c)  &\ge 11.  \label{contra}
\end{align}
This can only be satisfied if $z_a + z_c > 0$, but then, by (\ref{imply}), we get $y_a + y_c \le 1$ which makes inequality (\ref{contra}) unsatisfiable and, thus, leads to the desired contradiction. Therefore, such a Nash equilibrium strategy profile $\s$ cannot exist.
\end{proof}

\subsection{Price of Stability}

We start with a trivial observation:
If we use a uniform price function then the social optimum is a Nash equilibrium because all admissible strategy profiles have identical utilities.

\begin{observation}[Price of Stability Uniform]
    The price of stability for uniform price functions and therefore the whole charging game is $PoS = PoS_\textup{UNI}=1$.
\end{observation}
However, this price function is unnatural, as it does not occur in reality and there is also no incentive to actually implement the social optimum as a Nash equilibrium over any of the other admissible strategy profiles.

Next, we show that for some more natural price profiles, the social welfare of the best Nash equilibrium (if it exists) may be very low.
First, for ascending price profiles, we exploit the fact that holding one unit of electricity for the entire game is more profitable than constantly buying and selling.
However, this entirely blocks the battery from realizing larger flows.

\begin{theorem}[$T>2$ Price of Stability Ascending Prices]
\label{thm:ASCPoS}
    For $T$ time steps, the price of stability for ascending prices is at least
    $PoS_\textup{ASC}\geq\left\lfloor\frac{T}{2}\right\rfloor$.
\end{theorem}
\begin{proof}
    Let $G$ have one node $a$ that is also an agent.
    Let $d(a_i)$ be $+1$ for odd time steps $i$ and $-1$ for even time steps and we use an arbitrary ascending price profile $\p$.
    For $T=4$, this network is visualized in \Cref{fig:ASCPoS} with a price profile that yields $p_t = t$.
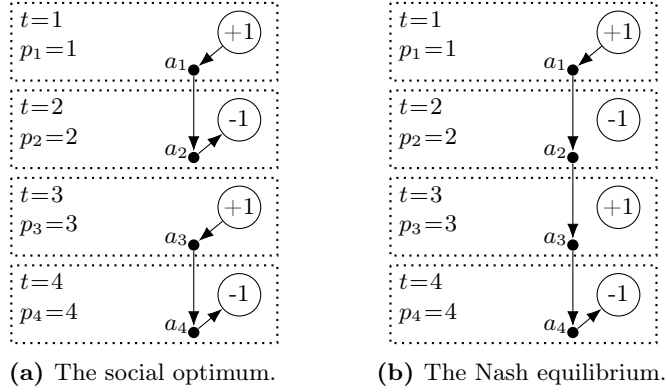
\begin{figure}[h]
\def\toffx{2.04}
\def\toffy{1.16}
\def\toffup{0.39}
\def\toffdown{0.64}
\def\toffleft{3}
\def\toffright{0.5}
    \centering
    \begin{subfigure}{0.4\columnwidth}
    \centering
    \begin{tikzpicture}
        \agent{a}{1}{+1}
        \agent{a}{2}{-1}
        \agent{a}{3}{+1}
        \agent{a}{4}{-1}
        
        \charge{a1}{}
        \keep{a1}{}
        \discharge{a2}{}
        
        \charge{a3}{}
        \keep{a3}{}
        \discharge{a4}{}
        
        \boxes{1}{1/1,2/2,3/3,4/4}
    \end{tikzpicture}
    \caption{The social optimum.}
    \end{subfigure}
    \begin{subfigure}{0.4\columnwidth}
    \centering
    \begin{tikzpicture}
        \agent{a}{1}{+1}
        \agent{a}{2}{-1}
        \agent{a}{3}{+1}
        \agent{a}{4}{-1}
        
        \charge{a1}{}
        \keep{a1}{}
        \keep{a2}{}
        \keep{a3}{}
        \discharge{a4}{}
        
        \boxes{1}{1/1,2/2,3/3,4/4}
    \end{tikzpicture}
    \caption{The Nash equilibrium.}
    \end{subfigure}
    \caption{(a) The strategy profile with the maximum flow value of $\left\lfloor\frac{T}{2}\right\rfloor$. (b) $G_\s$ of a Nash equilibrium $\s$ with flow value $1$. All capacities are $1$. The strategy profile in (a) is not a Nash equilibrium under any ascending price function.}
    \label{fig:ASCPoS}
\end{figure}

    The optimal flow has value $\left\lfloor\frac{T}{2}\right\rfloor$, where for each odd time step $t$, the unit of energy is sent to the consumer in the next time step $t+1$.
    In the only Nash equilibrium agent $a$ buys electricity in the first time step and sells in the last (or second to last for odd $T$) which yields a flow value of $1$.
    Thereby, we get a price of stability of at least $PoS_\textup{ASC}\geq\left\lfloor\frac{T}{2}\right\rfloor$.
\end{proof}

We also show that the natural class of price functions SIGN fails at having a low price of stability. Remember, a SIGN function has a negative price in time steps with an oversupply of electricity, a positive price in time steps with an undersupply and a price of $0$ otherwise.

\begin{theorem}[$T > 2$ Price of Stability Supply-Sign Prices]
    For a supply-sign price function, the price of stability is $PoS_\textup{SIGN}=\infty$.
\end{theorem}

\begin{proof}
We can exploit the fact that the prices only depend on the total demand in each time step, but not on the structure of the given network. Thus, using dummy agents to create dummy demand and supply in inaccessible parts of the network, we can create a bad example.
Consider the network and strategy profile in \Cref{fig:SSPoA} for some arbitrary $q,r > 0$.

\begin{figure}[h]
\def\toffx{2.04}
\def\toffy{1.44}
\def\toffup{0.4}
\def\toffdown{0.65}
\def\toffleft{3}
\def\toffright{0.5}
    \centering
    \begin{tikzpicture}
        \agent{a}{1}{$+1$}
        \agent{b}{1}{}
        \vertex{c}{1}{$-1$}
        \agent{a}{2}{}
        \agent{b}{2}{}
        \vertex{c}{2}{$+1$}
        \agent{a}{3}{}
        \agent{b}{3}{$-1$}
        \vertex{c}{3}{}

        \charge[cred, ultra thick]{a1}{$1$}
        \keep[cred, ultra thick]{a1}{$1$}
        \discharge[cred, ultra thick]{a2}{$1$}
        
        \keep{a2}{$1$}
        \keep{b1}{$1$}

        \path (a2) edge["1", pos=0.37, cred, ultra thick] (b2);
        
        \charge[cred, ultra thick]{b2}{$1$}
        \keep[cred, ultra thick]{b2}{$1$}
        \discharge[cred, ultra thick]{b3}{$1$}
        
        \boxes{3}{1/0,2/-q,3/r}
    \end{tikzpicture}
    \caption{In this instance, with a supply-sign pricing function, the strategy profile given by the thick red transaction edges produces a flow value of 1. However, it yields negative utility for agent $a$ and thus is not a Nash equilibrium.}
    \label{fig:SSPoA}
\end{figure}
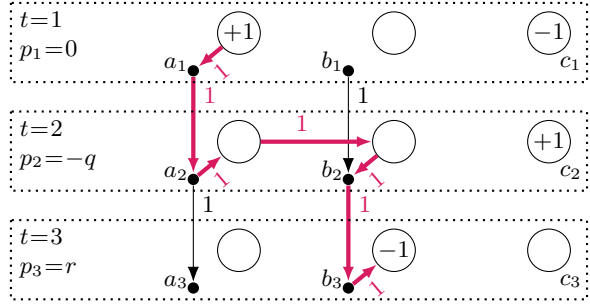

The given action profile has a flow of $1$ but results in a negative utility of $-q$ for agent $a$. The only equilibrium of this instance is $s_b=s_a=(0,0,0)$, resulting in an infinite price of stability.
\end{proof}


\subsection{Price of Anarchy}

Now, we show that if our game admits equilibria then they might be infinitely bad for any price profile.
To make matters worse, this result even holds for $k$-strong equilibria for any integer $k\leq T-2$ for $T$ time steps.
\begin{theorem}[$T>2$ Price of Anarchy, Low Cooperation, Any Prices]
\label{thm:cooperation}
    For $T>2$ and $k\leq T-2$, the $k$-strong price of anarchy is $PoA_{\textup{ANY}}^{k\textup{-strong}}=\infty$.
\end{theorem}

\begin{proof}
    Let $G$ be a directed path $(v^1, \dots, v^{T-1})$ with one agent per node each.
    Let $d(v^1_1)=+1$ and $d(v^{T-1}_T) = -1$ and let all battery capacities be $1$.
    The grid edge capacities are $1$ for all edges of the form $(v^i_{i+1}, v^{i+1}_{i+1})$ for $i \leq T-2$, otherwise they are $0$.
    See an example in \Cref{fig:strong-poa-construction}.
    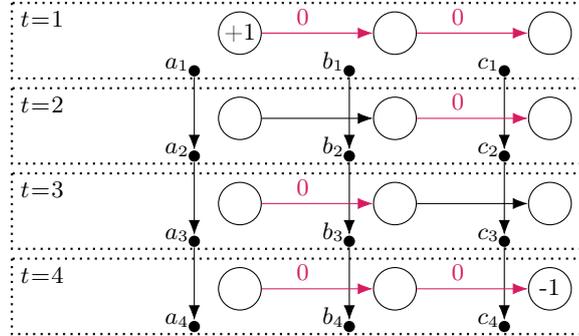
\begin{figure}[b!]
\def\toffx{2.04}
\def\toffy{1.13}
\def\toffup{0.4}
\def\toffdown{0.6}
\def\toffleft{3}
\def\toffright{0.5}
    \centering
    \begin{tikzpicture}
        \agent{a}{1}{+1}
        \agent{b}{1}{}
        \agent{c}{1}{}
        \agent{a}{2}{}
        \agent{b}{2}{}
        \agent{c}{2}{}
        \agent{a}{3}{}
        \agent{b}{3}{}
        \agent{c}{3}{}
        \agent{a}{4}{}
        \agent{b}{4}{}
        \agent{c}{4}{-1}

        \foreach \bat in {a,b,c} {
            \keep{\bat 1}{}
            \keep{\bat 2}{}
            \keep{\bat 3}{}
        }
        
        \path (a2) edge (b2);
        \path (b3) edge (c3);
        \path (b1) edge["0", pos=0.37, cred] (c1);
        \path (a4) edge["0", pos=0.37, cred] (b4);
        
        \path (a1) edge["0", pos=0.37, cred] (b1);
        \path (b2) edge["0", pos=0.37, cred] (c2);
        \path (a3) edge["0", pos=0.37, cred] (b3);
        \path (b4) edge["0", pos=0.37, cred] (c4);
    \boxes{3}{1/,2/,3/,4/}
    \end{tikzpicture}
    \caption{This is graph $G_\s$ for the graph in \Cref{thm:cooperation} with $T=4$ with the empty strategy profile $\s$, which is a Nash equilibrium.
    All black edges have capacity $1$, the red edges have capacity $0$.
    This graph needs cooperation of at least $T-1=3$ agents to escape from strategy profile $\s$ and achieve a higher flow value.
    Hence, for $1$- and $2$-strong equilibria, this example shows a price of anarchy of $\infty$.}
    \label{fig:strong-poa-construction}
\end{figure}
    
    Therefore, there is exactly one path between the only supply node and the only demand node in $G$ passing through a single battery edge of each agent.
    This means that the empty strategy profile is a $k$-strong equilibrium for any $k\leq T-2$.
    The social optimum has a flow of $1$ so the price of anarchy is infinite.
\end{proof}
While for low cooperation, i.e., for coalitions of up to $T-2$ agents, we have an infinite price of anarchy, this bound becomes finite if we raise the level of cooperation:
If at least $T-1$ agents cooperate in an instance with $T$ time steps, our upper bound for the price of anarchy is $T$.
This means that a game with more time steps needs more cooperation to have a finite price of anarchy.

To prove this statement, we look at the difference of the flow value in a Nash equilibrium and a maximum flow value in any strategy profile.
With such a high level of cooperation, any augmentation of the equilibrium flow lowers the flow on a battery edge.
This limits how much the equilibrium flow can be augmented.

\begin{theorem}[$T\geq 2$ Price of Anarchy, High Cooperation, Any Prices]
    \label{thm:poa-t-strong}
    For $T$ time steps and an ascending price function, the price of anarchy considering $(T-1)$-strong equilibria is
    $PoA_{\textup{ASC}}^{(T-1)\textup{-strong}} \leq T$.
\end{theorem}
\begin{proof}
Let $G$ be any electricity network graph with $T\geq2$ time steps and let strategy profile $\s$ be any $(T-1)$-strong equilibrium.
Let $f$ be a maximum flow in $G_\s$ and $f_\textup{OPT}$ be a maximum flow in~$G$.
We show that $\frac{|f_\textup{OPT}|}{|f|}\leq T$.

Consider the flow $f$ in the context of $\kappa$ instead of $\kappa_\s$, i.e., with the maximum battery capacities instead of the ones induced by strategy $\s$.
First, observe that the only way to augment flow $f$ with a single path is to decrease the flow on some battery edge $b_t$.
Otherwise, since we consider ascending prices, the single path yields an improving coalition of size at most $T-1$ because it can only use $T-1$ battery edges.
This would contradict $\s$ being a $(T-1)$-strong equilibrium.

Now consider $|f_\textup{diff}| = |f_\textup{OPT}| - |f|$, the difference of flow value between $f_\textup{OPT}$ and $f$, and decompose it into infinitesimally small augmentations, each along a single path.
Each augmentation decreases the flow along some battery edge $b_t$ in the original flow $f$ (which is why the respective agent would not agree to this joint strategy change).
Therefore, the flow difference is at most $|f_\textup{diff}| = |f|(T-1)$.
Hence,
\[
    \frac{|f_\textup{OPT}|}{|f|} \leq \frac{|f|(T-1) + |f|}{|f|} = T\text.\qedhere
\]
\end{proof}
Note that the upper bound of $T$ from \Cref{thm:poa-t-strong} is asymptotically tight for ascending price profiles, since for these price profiles, we have a lower bound on the price of stability and, therefore, also on the price of anarchy by \Cref{thm:ASCPoS}.

\section{Conclusion}
We introduce a novel model for selfish behavior in energy networks where strategic agents equipped with batteries interact and can use excess energy to profit. Our research uncovers that this selfish behavior severely impacts the obtained energy flow value, i.e., the utilization of excess energy due to the batteries is much lower than theoretically possible. Given that today we have millions of such battery-equipped agents connected to the power grid, this aspect of energy networks deserves further studies and should be included in other models. Our stylized model is just the first step in this research endeavor.

For our model, it remains open if there are special cases or natural pricing schemes with better price of anarchy/stability bounds and how to compute the best price profile for a given instance. 
Structural properties of the energy network might also play a role in this. 
Moreover, our model could be augmented with local prices that depend on the position in the network and not only on the time step.
This could even include incentives for self-consumption already present in some energy markets.
Furthermore, our agents precommit to a strategy knowing the full game in advance which could be changed to a more dynamic game with uncertainty.
Finally, experiments with real-world data might shed light on the actual impact of selfish battery charging decisions by algorithmic trading agents.

\bibliographystyle{splncs04}
\bibliography{references}

\end{document}